\documentclass[11pt]{article}%
\usepackage{amsfonts}
\usepackage{amsmath}
\usepackage{amssymb}
\usepackage{graphicx}%
\newtheorem{theorem}{Theorem}

\newtheorem{corollary}[theorem]{Corollary}

\newtheorem{proposition}[theorem]{Proposition}
\newtheorem{remark}[theorem]{Remark}

\newenvironment{proof}[1][Proof]{\noindent\textbf{#1.} }{\ \rule{0.5em}{0.5em}}
\begin{document}

\title{Remarks on a paper by Cordero and Nicola on Feichtinger's Wiener amalgam
spaces and the Schr\"{o}dinger equation}
\author{Maurice de Gosson\\\ Jacob's University, Bremen}
\maketitle

\begin{abstract}
We derive some consequences of very recent results of Cordero and Nicola on
the metaplectic representation, the Wiener amalgam spaces, (whose definition
is due to Feichtinger), and their applications to the regularity of the
solutions of Schr\"{o}dinger equation with quadratic Weyl symbol. We do not
however discuss the validity of Cordero and Nicola's claims.

\end{abstract}

\section{Introduction}

There are very few results in the literature about the regularity of the
solutions of Schr\"{o}dinger's equation in terms of the Wiener amalgam spaces
introduced by Feichtinger in the early 1980s. Very recently, Cordero and
Nicola \cite{cono2} have proposed such a study for the quantum oscillator
(their study actually applies to larger classes of quadratic Hamiltonians as
well); their study is based on properties of the action of the metaplectic
group on Feichtinger's spaces they prove.

The aim of this short Note is to derive a few consequences of Cordero and
Nicola's results by using the complete properties of the metaplectic operators
they invoke. We do not, however, discuss the validity of their results.

\section{The metaplectic group}

Let $\operatorname*{Sp}(n)$ be the symplectic group of the space
$\mathbb{R}^{2n}$ equipped with the standard symplectic structure
$\sigma=dp\wedge dx$. It is a connected classical Lie group, contractible to
the unitary group $U(n,\mathbb{C})$. We thus have $\pi_{1}[\operatorname*{Sp}%
(n)]=(Z,+)$ and $\operatorname*{Sp}(n)$ therefore has covering groups of all
orders. It turns out that the double cover $\operatorname*{Sp}_{2}(n)$ can be
faithfully represented by a group of unitary operators on $L^{2}%
(\mathbb{R}^{n})$. That group is the metaplectic group $\operatorname*{Mp}%
(n)$. Since the projection $\pi:\operatorname*{Mp}(n)\longrightarrow
\operatorname*{Sp}(n)$ is two-to-one, one associates to each $\mathcal{A}%
\in\operatorname*{Sp}(n)$ \emph{two} operators $\pm\mu(\mathcal{A)}%
\in\operatorname*{Mp}(n)$. In particular, if%
\begin{equation}
\mathcal{A=}%
\begin{bmatrix}
A & B\\
C & D
\end{bmatrix}
\text{ \ , \ }\det B\neq0 \label{free}%
\end{equation}
then
\begin{equation}
\pm\mu(\mathcal{A})f(x)=(2\pi\hbar)^{-n/2}i^{m-n/2}|\det B|^{-1/2}\int
e^{^{\frac{i}{\hbar}}W(x,x^{\prime})}f(x^{\prime})d^{n}x^{\prime} \label{sw1}%
\end{equation}
where $m$ ("the Maslov index \cite{Wiley,Leray}") corresponds to a choice of
$\arg\det B$ and
\[
W(x,x^{\prime})=\tfrac{1}{2}DB^{-1}x\cdot x-B^{-1}x\cdot x^{\prime}+\tfrac
{1}{2}B^{-1}Ax^{\prime}\cdot x^{\prime}.
\]

It turns out that \emph{each} $\mu(\mathcal{A)}$ can be written as a product
of exactly \emph{two }operators of the type above. In fact, writing $\mu
_{W,m}(\mathcal{A})$ for the operator (\ref{sw1}) one the Maslov index $m$ has
been chosen:

\begin{proposition}
\label{un}For every operator $\mu(\mathcal{A})\in\operatorname*{Mp}(n)$ there
exist quadratic forms $W$, $W^{\prime}$ and integers (modulo 4) $m$ and
$m^{\prime}$ such that
\[
\mu(\mathcal{A})=\mu_{W,m}(\mathcal{A})\mu_{W^{\prime},m^{\prime}}%
(\mathcal{A})\text{. }%
\]
(This factorization is of course not unique).
\end{proposition}

\begin{proof}
See de Gosson \cite{AIF,Cocycles,Wiley,ICP,Birk}.
\end{proof}

\begin{remark}
A similar result holds on the matrix level: every $\mathcal{A}\in
\operatorname*{Sp}(n)$ can be written as a product of two symplectic matrices
(\ref{free}).
\end{remark}

\section{On Cordero and Nicola's estimate}

In \cite{cono2} Cordero and Nicola prove the following continuity result for
metaplectic operators on Wiener amalgam spaces (Theorem 4.1, formula (25), p. 17):

\begin{proposition}
\label{procono}Let $\mathcal{A}\in\operatorname*{Sp}(n)$ be as in
(\ref{free}). Then, for $1\leq p,q\leq\infty$,
\begin{equation}
||\mu(\mathcal{A})f||_{W(\mathcal{F}L^{p},L^{q})}\leq\alpha(\mathcal{A}%
,p,q)||f||_{W(\mathcal{F}L^{q},L^{p})}\text{.} \label{conest}%
\end{equation}
where $\alpha(\mathcal{A},p,q)>0$
\end{proposition}

We claim that:

\begin{proposition}
\label{two}Let $\mathcal{A}\in\operatorname*{Sp}(n)$. Then, for $1\leq
p,q\leq\infty$,%
\begin{equation}
||\mu(\mathcal{A})f||_{W(\mathcal{F}L^{p},L^{q})}\leq C_{\mathcal{A}%
,p,q}||f||_{W(\mathcal{F}L^{p},L^{q})} \label{monest}%
\end{equation}
where $C_{\mathcal{A},p,q}>0$ only depends on $\mathcal{A},p,q$.
\end{proposition}

\begin{proof}
It immediately follows from (\ref{conest}) using Proposition \ref{un}.
\end{proof}

Thus:

\begin{corollary}
Metaplectic operators are continuous operators on the Wiener amalgam spaces
$W(\mathcal{F}L^{p},L^{q})$, $1\leq p,q\leq\infty$.
\end{corollary}

\begin{remark}
This result can actually be obtained in a simpler way, using different methods
(the Weyl representation of metaplectic operators studied in de Gosson
\cite{lett}).
\end{remark}

\section{Schr\"{o}dinger's equation}

Consider the time-dependent Schr\"{o}dinger equation%
\begin{equation}
i\hbar\frac{\partial}{\partial t}f(x,t)=H_{\text{Weyl}}f(x,t)\text{ \ ,
\ }f(x,0)=f_{0}(x) \label{sch}%
\end{equation}
where the operator $H_{\text{Weyl}}$ is the partial differential operator with
Weyl symbol%
\[
H(x,p)=\frac{1}{2}(x,p)M(x,p)^{T}%
\]
where $M$ is a real symmetric $2n\times2n$ matrix. The corresponding Hamilton
equations of motion%
\[
\frac{dx}{dt}=\frac{\partial H}{\partial p}\text{ \ , \ }\frac{dp}{dt}%
=-\frac{\partial H}{\partial x}%
\]
are then linear; the Hamiltonian flow determined by these equations is thus a
one-parameter subgroup $(\mathcal{A}_{t})$ of $\operatorname*{Sp}(n)$. In view
of the lifting theorem of algebraic topology (see e.g. \cite{GS2}), there
exists a \emph{unique} one-parameter group $(\mu(\mathcal{A}_{t}))$ of
$\operatorname*{Mp}(n)$ covering $(\mathcal{A}_{t})$.

\begin{proposition}
\label{deux}Assume $f_{0}\in L^{2}(\mathbb{R}^{n})$ and set $f(x,t)=\mu
(\mathcal{A}_{t})f_{0}(x)$. The function $f$ is the (unique) solution of
Schr\"{o}dinger's equation (\ref{sch}).
\end{proposition}

\begin{proof}
See \cite{GS2,Birk}.
\end{proof}

An immediate consequence of this classical result is:

\begin{corollary}
Assume that $f_{0}\in W(\mathcal{F}L^{p},L^{q}).$ Then $f(,\cdot t)\in
W(\mathcal{F}L^{p},L^{q})$ for all $t$.
\end{corollary}

\begin{proof}
It immediately follows from Proposition \ref{deux} using the estimate
(\ref{monest}) in Proposition \ref{two}.
\end{proof}

\section{Conclusion}

The usefulness of Feichtinger's Wiener amalgam spaces in quantum mechanics is
obvious, as is the usefulness of Feichtinger's algebra \cite{HF1}, which
deserves to be investigated separately. We will come back in a forthcoming
paper \cite{GL} to a precise study of then regularity of the solutions to
Schr\"{o}dinger's equation for quite arbitrary Hamiltonians (i.e. not
necessarily associated to a quadratic Hamiltonian) in terms of these spaces.

\end{document}